\documentclass{llncs}
\usepackage{enumerate}
\usepackage{amssymb}
\usepackage{xypic}
\title{LIFO-search on digraphs:\\A searching game for cycle-rank\thanks{This work resulted from discussions during Dagstuhl Seminar 11071 on Graph Searching, Theory and Applications}}

\author{Paul Hunter\thanks{Supported by Model Checking Real-Time Systems project, EPSRC grant ref.~BLRQEK}}

\institute{Department of Computer Science, University of Oxford\\\email{paul.hunter@cs.ox.ac.uk}}

\spnewtheorem*{mainthm}{Main Theorem}{\bfseries}{\itshape}

\newcommand{\letters}[1]{\{\!|#1|\!\}}
\newcommand{\st}{\mathrel{:}}
\newcommand{\cycr}[1]{\textrm{cr}(#1)}
\newcommand{\lifo}[1]{\textrm{LIFO}^{#1}}
\newcommand{\sstat}[1]{\textrm{SS}^{#1}}
\newcommand{\pow}[1]{\mathcal{P}(#1)}
\newcommand{\var}{\texttt{\bf gv}}
\newcommand{\ivar}{\texttt{i}}
\newcommand{\vvar}{\texttt{v}}
\newcommand{\iscvar}{\texttt{isc}}
\newcommand{\vscvar}{\texttt{vsc}}
\newcommand{\mvar}{\texttt{m}}

\begin{document}
\maketitle

\begin{abstract}
We consider the extension of the last-in-first-out graph searching game of Giannopoulou and  Thilikos to digraphs.  We show that all common variations of the game require the same number of searchers, and the minimal number of searchers required is one more than the cycle-rank of the digraph.  We also obtain a tight duality theorem, giving a precise min-max characterization of obstructions for cycle-rank.
\end{abstract}
\section{Introduction}
%

Graph searching games are increasingly becoming a popular way to characterize, and even define, practical graph parameters.   There are many advantages to a characterization by graph searching games:  it provides a useful intuition which can assist in constructing more general or more specific parameters; it gives insights into relations with other, similarly characterized parameters; 
and it is particularly useful from an algorithmic perspective as many parameters associated with such games are both structurally robust and efficiently computable.  


One of the most common graph searching games is the node-search game.  In this game several searchers and one fugitive occupy vertices of the graph and make simultaneous moves.  The (omniscient) fugitive moves along searcher-free paths of arbitrary length whereas the searchers' movements are not constrained by the topology of the graph.  The goal of the game is to minimize the number of searchers required to capture the fugitive by cornering him in some part of the graph and placing a searcher on the same vertex.  This game has been extensively studied~\cite{DKT97} and several important graph parameters such as treewidth~\cite{ST93}, pathwidth~\cite{KP86}, and tree-depth~\cite{NdM06} can be characterized by natural variants of this game.  
One variation frequently used, indeed the one which separates treewidth and pathwidth, is whether the location of the fugitive is known or unknown to the searchers.  
Another common variation is whether the searchers use a monotone or a non-monotone searching strategy.  Monotone search strategies lead to algorithmically useful decompositions, whereas non-monotone strategies are more robust under graph operations and hence reflect structural properties, so showing that monotone strategies require no more searchers than non-monotone strategies is an important and common question in the area.
Whilst node-search games on undirected graphs tend to enjoy monotonicity~\cite{BS91,ST93,LaP93}, on digraphs the situation is much less clear~\cite{Bar05,Adl05,KO07}.

Node-search games naturally extend to digraphs, however, in the translation another variation arises depending on how one views the constraints on the movement of the fugitive.  One interpretation is that in the undirected case the fugitive moves along paths, so the natural translation would be to have the fugitive move along directed paths.  Another view is that the fugitive moves to some other vertex in the same connected component, and here the natural translation would be to have the fugitive move within the same strongly connected component.  Both interpretations have been studied in the literature, the former giving characterizations of parameters such as DAG-width~\cite{BDHK06,Obd06} and directed pathwidth~\cite{Bar05} and the latter giving a characterization of directed treewidth~\cite{JRST01}.  

In~\cite{GT11}, Giannopoulou and Thilikos define a variant of the node-search game in which only the most recently placed searchers may be removed; that is, the searchers must move in a last-in-first-out (LIFO) manner.  They show that the visibility of the fugitive is not relevant to the minimum number of searchers required, the game is monotone, and that it characterizes tree-depth.  In this paper we consider the extension of this game to digraphs.

We generalize the results of Giannopoulou and Thilikos by showing that
the minimum number of searchers required to capture a fugitive on a digraph with a LIFO-search is independent of:
\begin{itemize}
\item Whether the fugitive is invisible or visible,
\item Whether the searchers use a monotone or non-monotone search, and
\item Whether the fugitive is restricted to moving in searcher-free strongly connected sets or along searcher-free directed paths.
\end{itemize}
This result is somewhat surprising: in the standard node-search game these options give rise to quite different parameters~\cite{Bar05,BDHK06,KO07}.

We show that on digraphs the LIFO-search game also characterizes a pre-existing measure, cycle-rank -- a generalization of tree-depth to digraphs (though as the definition of cycle-rank predates tree-depth by several decades, it is perhaps more correct to say that tree-depth is an analogue of cycle-rank on undirected graphs).  The cycle-rank of a digraph is an important parameter relating digraph complexity to other areas such as regular language complexity and asymmetric matrix factorization.  It was defined by Eggan~\cite{Egg63}, where it was shown to be a critical parameter for determining the star-height of regular languages,  
and interest in it as an important digraph parameter, especially from an algorithmic perspective, has recently been rekindled by the success of tree-depth~\cite{EL08,Gru08,GHKLOR09}. 

It is well known that tree-depth can also be characterized by a node-search game where a visible fugitive plays against searchers that are only placed and never moved~\cite{GHKLOR09}.  In that paper, Ganian et al.~considered one extension of this game to digraphs.  Here we consider the other natural extension, where the visible fugitive moves in strongly connected sets, and show that it also characterizes cycle-rank.

Our final result uses these graph searching characterizations to define a dual parameter that characterizes structural obstructions for cycle-rank.  We consider two obstructions, motivated by the shelters of~\cite{GT11} and the havens of~\cite{JRST01}, that define simplified strategies for the fugitive. The game characterization then implies that these structural features are necessarily present when the cycle-rank of a graph is large.  By showing that such strategies are also sufficient for the fugitive, we obtain a rare instance of an exact min-max theorem relating digraph parameters.

The results of this paper can be summarized with the following characterizations of cycle-rank.
\begin{mainthm}\label{thm:main2}
Let $G$ be a digraph, and $k$ a positive integer.  The following are equivalent:
\begin{enumerate}[(i) ]
\item $G$ has cycle-rank $\leq k-1$,
\item On $G$, $k$ searchers can capture a fugitive with a LIFO-search strategy,
\item On $G$, $k$ searchers can capture a visible fugitive restricted to moving in strongly connected sets with a searcher-stationary search strategy,
\item $G$ has no LIFO-haven of order $> k$, and
\item $G$ has no strong shelter of thickness $>k$.
\end{enumerate}
\end{mainthm}

The paper is organised as follows.  In Section~\ref{sec:prelim} we recall the definitions and notation that we use throughout the paper.  In Section~\ref{sec:game} we define the LIFO-search and searcher-stationary games and show that they characterize cycle-rank.  In Section~\ref{sec:min-max} we prove the min-max theorem for cycle-rank, and in Section~\ref{sec:conc} we conclude with a discussion on further research and open problems.

\section{Preliminaries}\label{sec:prelim}
All (di)graphs in this paper are finite, simple, directed and without self-loops, although the results readily extend to multigraphs with self-loops.  For simplicity, we also assume that all digraphs contain at least one vertex unless explicitly mentioned.  We use standard notation and terminology, in particular $V(G)$ and $E(G)$ denote the sets of vertices and edges respectively of a digraph $G$ and between digraphs, $\subseteq$ denotes the subgraph relation.  We will often interchange an induced subgraph with the set of vertices which defines it, in particular strongly connected sets of vertices are sets of vertices that induce a strongly connected subgraph, and we will often view strongly connected components as sets of vertices.  Given a digraph $G$ and a set of vertices $X \subseteq V(G)$, we use $G \setminus X$ to denote the subgraph of $G$ induced by $V(G) \setminus X$.  An \emph{initial component} of a digraph $G$ is a strongly connected component $C$ with no edges from $G\setminus C$ to $C$.  $H \subseteq G$ is \emph{successor-closed} if there are no edges in $G$ from $H$ to $G \setminus H$.

Given a finite set $V$, we use $V^*$ to denote the set of finite words over $V$, and $V^{< k}$ to denote the set of words over $V$ of length $< k$.  We use $\epsilon$ to denote the empty word and $\cdot$ or juxtaposition to denote concatenation.  For $X,Y \in V^*$ we write $X \preceq Y$ if $X$ is a prefix of $Y$, that is if there exists a word $Z \in V^*$ such that $Y = X\cdot Z$.  For $X=a_1a_2\cdots a_n\in V^*$, we use $|X|$ to denote the length of $X$, and $\letters{X}$ to denote the set $\{a_1, a_2, \ldots, a_n\}$.  Given two sets $A$ and $B$ we use $A \Delta B$ to denote their symmetric difference, that is $A \Delta B = (A \cup B) \setminus (A \cap B)$.  Given a set $\mathcal{S} \subseteq \mathcal{P}(V)$ of subsets of $V$, a \emph{$\subseteq$-chain} is a subset $\{X_1,\ldots, X_n\} \subseteq \mathcal{S}$ such that $X_1 \subseteq X_2 \subseteq \cdots \subseteq X_n$.  If there is no $Y \in \mathcal{S}$ such that $Y \subset X_1$, $X_i \subset Y \subset X_{i+1}$ for some $i$, or $X_n \subset Y$, then $\{X_1, \ldots, X_n\}$ is a \emph{maximal $\subseteq$-chain}.

The \emph{cycle-rank} of a digraph $G$, $\cycr{G}$, is defined as follows:
\begin{itemize}
\item If $G$ is acyclic then $\cycr{G}=0$.
\item If $G$ is strongly connected then $\cycr{G} = 1 + \min_{v \in V(G)} \cycr{G \setminus \{v\}}$.
\item Otherwise $\cycr{G} = \max_H \cycr{H}$ where the maximum is taken over all strongly connected components $H$ of $G$.
\end{itemize}

\section{Searching games for cycle-rank}\label{sec:game}
We begin by formally defining the LIFO-search game, and its variants, for digraphs.  Each variation of the LIFO-search game gives rise to a digraph parameter corresponding to the minimum number of searchers required to capture the fugitive under the given restrictions.  The main result of this section is that for any digraph all these parameters are equal.  Furthermore, we show they are all equal to
one more than the cycle-rank of the digraph.

\subsection{LIFO-search for digraphs}
In summary, for the graph searching game in which we are interested the fugitive can run along searcher-free directed paths of any length, the searchers can move to any vertex in the graph, and the fugitive moves whilst the searchers are relocating.
The only restriction we place on the searchers is that only the most recently placed searchers may be removed.   If a searcher is placed on the fugitive then he is captured and the searchers win,  otherwise the fugitive wins.
The goal is to determine the minimum number of searchers required to capture the fugitive.  For simplicity we assume that each searcher move consists of either placing or removing one searcher and observe that this does not affect the minimum number of searchers required to capture the fugitive.
The variants we are primarily interested in are whether the searchers use a monotone or a non-monotone strategy, whether the fugitive is visible or invisible, and whether or not the fugitive must stay within the same strongly connected component when he is moving.   As our fundamental definitions are dependent on these latter two options, we define four \emph{game variants}: $\ivar, \iscvar, \vvar, \vscvar$, corresponding to the visibility of the fugitive and whether he is constrained to moving within strongly connected components, and parameterize our definitions by these variants.

Let us fix a digraph $G$.  A position in a LIFO-search on $G$ is a pair $(X,R)$ where $X \in V(G)^*$ and $R$ is a (possibly empty) induced subgraph of $G \setminus \letters{X}$.  Intuitively $X$ represents the position and ordered placement of the searchers and $R$ represents the part of $G$ that the fugitive can reach (in the visible case) or the set of vertices where he might possibly be located (in the invisible case).
We say a position $(X,R)$ is an \emph{$\ivar$-position} if $R$ is successor-closed;  an \emph{$\iscvar$-position} if it is a union of strongly connected components of $G \setminus \letters{X}$; a \emph{$\vvar$-position} if $R$ is successor-closed and has a unique initial component; and a \emph{$\vscvar$-position} if $R$ is a strongly connected component of $G \setminus \letters{X}$.
   
To reflect how the game transitions to a new position during a round of the game we say, for $\var \in \{\ivar,\iscvar,\vvar,\vscvar\}$, a $\var$-position $(X',R')$ is a \emph{$\var$-successor} of $(X,R)$ if
either $X \preceq X'$ or $X' \preceq X$, with $|\letters{X}\Delta\letters{X'}| = 1$, and 
\begin{itemize}
\item (for $\var \in \{\ivar,\vvar\}$) For every $v' \in V(R')$ there is a $v \in V(R)$ and a directed path in $G \setminus (\letters{X} \cap \letters{X'})$ from $v$ to $v'$, or
\item (for $\var \in \{\iscvar,\vscvar\}$) For every $v' \in V(R')$ there is a $v \in V(R)$ such that $v$ and $v'$ are contained in the same strongly connected component of $G \setminus (\letters{X} \cap \letters{X'})$.
\end{itemize}
Ideally we would like to assume games start from $(\epsilon, G)$, however in the visible variants of the game this might not be a legitimate position.  Thus, for $\var \in \{\vvar,\vscvar\}$, if $(\epsilon, G)$ is not a $\var$-position we include it as a special case, and set as its $\var$-successors all $\var$-positions of the form $(\epsilon, R)$.  
We observe that in all variants, the successor relation is monotone in the sense that if $(X,R)$ and $(X,S)$ are positions with $S \subseteq R$ and $(X',S')$ is a successor of $(X,S)$, then there is a successor $(X',R')$ of $(X,R)$ with $S' \subseteq R'$.


For $\var \in \{\ivar,\iscvar,\vvar,\vscvar\}$, a \emph{($\var$-LIFO-)search} in a digraph $G$ from $\var$-position $(X,R)$ is a (finite or infinite) sequence of $\var$-positions $(X,R) =(X_0,R_0),$ $(X_1,R_1),\ldots$ where for all $i \geq 0$, $(X_{i+1},R_{i+1})$ is a $\var$-successor of $(X_i,R_i)$.   A LIFO-search is \emph{complete} if 
either $R_n = \emptyset$ for some $n$, or it is infinite.  We observe that if $R_n = \emptyset$, then $R_{n'} = \emptyset$ for all $n' \geq n$. 

We say a complete LIFO-search is \emph{winning for the searchers} if $R_n = \emptyset$ for some $n$, otherwise it is winning for the fugitive.   A complete LIFO-search from $(\epsilon, G)$ is \emph{monotone} if $R_{i+1} \subseteq R_i$ for all $i$; it is \emph{searcher-stationary} if $X_{i} \preceq X_{i+1}$ for all $i$ where $R_i \neq \emptyset$; and it \emph{uses at most $k$ searchers} if $|X_i| \leq k$ for all $i$.  

Whilst a complete LIFO-search from $(\epsilon, G)$ describes a single run of the game, we are more interested in the cases where one of the players (particularly the searchers) can always force a win, no matter what the other player chooses to do.  For this, we introduce the notion of a strategy.  For $\var \in \{\ivar,\iscvar,\vvar,\vscvar\}$, a \emph{(searcher) $\var$-strategy} is a (partial\footnote{A strategy need only be defined for all positions $(X,R)$ that can be reached from $(\epsilon, G)$ in a LIFO-search consistent with the strategy.  However, as this definition is somewhat circular, we assume strategies are total.}) function $\sigma$ from the set of all $\var$-positions to $V(G)^*$ such that for all $(X,R)$, $\sigma(X,R)$ is the first component of a $\var$-successor of $(X,R)$; so with the possible exception of $(X,R) = (\epsilon, G)$, either $\sigma(X,R) \preceq X$ or $X \preceq \sigma(X,R)$.  A $\var$-LIFO-search $(X_0,R_0),(X_1,R_1),\ldots$ is \emph{consistent} with a $\var$-strategy $\sigma$ if $X_{i+1} = \sigma(X_i,R_i)$ for all $i \geq 0$.  A strategy $\sigma$ is winning from $(X,R)$ if all complete LIFO-searches from $(X,R)$ consistent with $\sigma$ are winning for the searchers.  Likewise, a strategy is monotone (searcher-stationary, uses at most $k$ searchers) if all consistent complete LIFO-searches from $(\epsilon, G)$ are monotone (searcher-stationary, use at most $k$ searchers respectively).  We say $k$ searchers can capture a fugitive on $G$ in the $\var$-game with a (monotone) LIFO-search strategy if there is a (monotone) $\var$-strategy that uses at most $k$ searchers and is winning from $(\epsilon, G)$. 

For $\var \in \{\ivar,\iscvar, \vvar, \vscvar\}$, we define the (monotone) $\var$-LIFO-search number of $G$, $\lifo{\var}(G)$ ($\lifo{\mvar\var}(G)$), as the minimum $k$ for which there is a (monotone) winning $\var$-strategy that uses at most $k$ searchers.  We also define the visible, strongly connected, searcher-stationary search number of $G$, $\sstat{\vscvar}(G)$ as the minimum $k$ for which there is a searcher-stationary winning $\vscvar$-strategy that uses at most $k$ searchers.  

In Section~\ref{sec:min-max} we will also consider fugitive $\var$-strategies: a partial function $\rho$ from $V(G)^* \times \pow{G} \times V(G)^*$ to induced subgraphs of $G$, defined for $(X,R,X')$ if $(X,R)$ is a $\var$-position and $X'$ is the first component of a $\var$-successor of $(X,R)$.  A LIFO-search $(X_0,R_0),(X_1,R_1),\ldots$ is \emph{consistent} with a fugitive $\var$-strategy $\rho$ if $R_{i+1} = \rho(X_i,R_i,X_{i+1})$ for all $i \geq 0$, and a fugitive strategy is winning if all consistent complete LIFO-searches are winning for the fugitive.   In this section, a strategy will always refer to a searcher strategy.



\subsection{Relating the digraph searching parameters}
We observe that in all game variants, a strategy that is winning from $(X,R)$ can be used to define a strategy that is winning from $(X,R')$ for any $R' \subseteq R$: the searchers can play as if the fugitive is located in the larger space; and from the monotonicity of the successor relation, the assumption that the actual set of locations of the fugitive is a subset of the assumed set of locations remains invariant.  
One consequence is that a winning strategy on $G$ defines a winning strategy on any subgraph of $G$, so the search numbers we have defined are monotone with respect to the subgraph relation.

\begin{proposition}\label{prop:subgraph}
Let $G$ be a digraph and $G'$ a subgraph of $G$.  Then:
\begin{itemize}
\item $\sstat{\vscvar}(G')\leq \sstat{\vscvar}(G)$, and
\item $\lifo{\var}(G') \leq \lifo{\var}(G)$ for $\var \in \{\ivar,\iscvar, \vvar, \vscvar, \mvar\ivar,\mvar\iscvar, \mvar\vvar, \mvar\vscvar\}$.
\end{itemize}
\end{proposition}

Another consequence is that a winning strategy in the invisible fugitive variant defines a winning strategy when the fugitive is visible; and a winning strategy when the fugitive is not constrained to moving within strongly connected components defines a winning strategy when he is.  This corresponds to our intuition of the fugitive being more (or less) restricted.  Also, in all game variants, a monotone winning strategy is clearly a winning strategy, and because a searcher-stationary LIFO-search is monotone, a winning searcher-stationary strategy is a monotone winning strategy.  These observations yield several inequalities between the search numbers defined above.  For example $\lifo{\vscvar}(G) \leq \lifo{\mvar\ivar}(G)$ as any winning monotone $\ivar$-strategy is also a winning $\vscvar$-strategy.  The full set of these relationships is shown in a Hasse diagram in Figure~\ref{fig:rel}, with the larger measures towards the top.

\begin{figure}
\[\xymatrix@R=20pt@C=20pt{
&\lifo{\mvar\ivar}(G)\ar@{-}[dl]\ar@{-}[d]\ar@{-}[dr]\\
\lifo{\ivar}(G)\ar@{-}[d]\ar@{-}[dr]&\lifo{\mvar\vvar}(G)\ar@{-}[dl]\ar@{-}[dr]&\lifo{\mvar\iscvar}(G)\ar@{-}[d]\ar@{-}[dl]&\sstat{\vscvar}(G)\ar@{-}[dl]\\
\lifo{\vvar}(G)\ar@{-}[dr]&\lifo{\iscvar}(G)\ar@{-}[d]&\lifo{\mvar\vscvar}(G)\ar@{-}[dl]\\
&\lifo{\vscvar}(G)
 }\]
\caption{Trivial relations between digraph searching parameters}\label{fig:rel}
\end{figure}
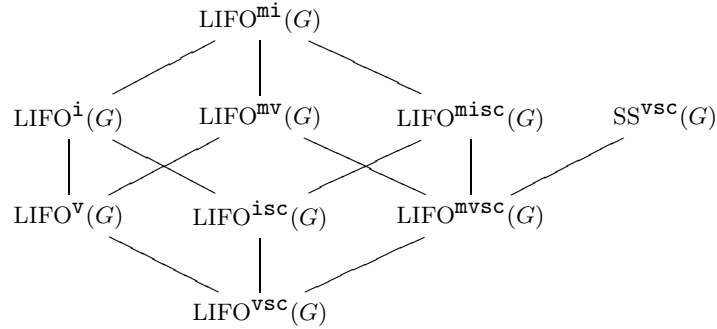

The main result of this section is that all these digraph parameters are equal to one more than cycle-rank.
\begin{theorem}\label{thm:game-char}
For any digraph $G$:
\[\begin{array}{rl}
	 1+ \cycr{G} &= \lifo{\mvar\ivar}(G) = \lifo{\ivar}(G) = \lifo{\mvar\iscvar}(G) = \lifo{\iscvar}(G) \\
	&= \lifo{\mvar\vvar}(G) = \lifo{\vvar}(G) = \lifo{\mvar\vscvar}(G) = \lifo{\vscvar}(G)\\
	&=\sstat{\vscvar}(G).
\end{array}\]
\end{theorem}
\begin{proof}
From the above observations, to prove Theorem~\ref{thm:game-char} it is sufficient to prove the following three inequalities:
\begin{enumerate}[(1) ]
\item $\lifo{\vscvar}(G) \geq \sstat{\vscvar}(G)$,
\item $\sstat{\vscvar}(G) \geq 1 + \cycr{G}$, and
\item $1 + \cycr{G} \geq \lifo{\mvar\ivar}(G)$.
\end{enumerate}
\end{proof}
These are established with the following series of lemmas.
\begin{lemma}  
For any digraph $G$, $\lifo{\vscvar}(G) \geq \sstat{\vscvar}(G)$.
\end{lemma}
\begin{proof}
We show that if a $\vscvar$-strategy is not searcher-stationary then it is not a winning strategy from $(\epsilon, G)$.  The result then follows as this implies every winning $\vscvar$-strategy is searcher-stationary.
Let $\sigma$ be a $\vscvar$-strategy, and suppose $(X_0, R_0), (X_1, R_1), \ldots$ is a complete $\vscvar$-LIFO-search from $(X_0,R_0) = (\epsilon, G)$ consistent with $\sigma$ which is not searcher-stationary.  
Let $j$ be the least index such that $X_{j} \succeq X_{j+1}$ and $R_{j} \neq \emptyset$.  As $X_0 = \epsilon$, there exists $i < j$ such that $X_i = X_{j+1}$.  By the minimality of $j$, and the assumption that we only place or remove one searcher in each round, $i = j-1$.   As $X_{j-1} \preceq X_{j}$, $R_{j} \subseteq R_{j-1}$, and as $X_{j+1} \preceq X_{j}$, $R_{j} \subseteq R_{j+1}$.  As $R_{j} \neq \emptyset$, it follows that $R_{j-1}$ and $R_{j+1}$ are the same strongly connected component of $G \setminus \letters{X_{j-1}}$.  Thus $(X_{j-1},R_{j-1})$ is a $\vscvar$-successor of $(X_{j},R_{j})$.  As $\sigma(X_j,R_j) = X_{j+1} = X_{j-1}$, it follows that $(X_0, R_0), (X_1, R_1), \ldots (X_{j-1},R_{j-1}),(X_{j},R_{j}),(X_{j-1},R_{j-1}),(X_{j},R_{j}), \ldots$ is an infinite, and hence complete,  $\vscvar$-LIFO-search (from $(\epsilon, G)$) consistent with $\sigma$.  As $R_i \neq \emptyset$ for all $i \geq 0$, the LIFO-search is not winning for the searchers.  Thus $\sigma$ is not a winning strategy.
\end{proof}

\begin{lemma}  
For any digraph $G$, $\sstat{\vscvar}(G) \geq 1+\cycr{G}$.
\end{lemma}
\begin{proof}
We prove this by induction on $|V(G)|$.

If $|V(G)|=1$, then $\sstat{\vscvar}(G) = 1 = 1 + \cycr{G}$. 

Now suppose $\sstat{\vscvar}(G') \geq 1+\cycr{G'}$ for all digraphs $G'$ with $|V(G')| < |V(G)|$.
We first consider the case when $G$ is not strongly connected.  From Proposition~\ref{prop:subgraph}, $\sstat{\vscvar}(G) \geq \max_H \sstat{\vscvar}(H)$ where the maximum is taken over all strongly connected components $H$ of $G$.  As $G$ is not strongly connected, $|V(H)| < |V(G)|$ for all strongly connected components $H$ of $G$.  Therefore, by the induction hypothesis
\begin{eqnarray*}
\sstat{\vscvar}(G) &\geq& \max_H \:\sstat{\vscvar}(H)\\
&\geq& \max_H \:(1+ \cycr{H})\\
&=& 1 + \cycr{G}.
\end{eqnarray*}
Now suppose $G$ is strongly connected.  Let $\sigma$ be a winning searcher-stationary $\vscvar$-strategy which uses $\sstat{\vscvar}(G)$ searchers.  As $(\epsilon, G)$ is a legitimate $\vscvar$-position, if $(X,R)$ is a  $\vscvar$-successor of $(\epsilon, G)$ then $|X| = 1$.  Thus $|\sigma(\epsilon, G)| = 1$.  Let $\sigma(\epsilon,G) = v_0$.  As $\sigma$ is a searcher-stationary strategy which uses the minimal number of searchers, it follows that $\sstat{\vscvar}(G\setminus \{v_0\}) = \sstat{\vscvar}(G) - 1$.  Thus, by the induction hypothesis,
\begin{eqnarray*}
\sstat{\vscvar}(G) & = & \sstat{\vscvar}(G\setminus \{v_0\}) + 1\\
&\geq& (1+ \cycr{G\setminus \{v_0\}})+1\\
&\geq & (1 + \min_{v \in V(G)}\cycr{G\setminus \{v\}})+1\\
& = & 1 + \cycr{G}.
\end{eqnarray*}
\end{proof}

\begin{lemma}  
For any digraph $G$, $1+\cycr{G} \geq \lifo{\mvar\ivar}(G)$.
\end{lemma}
\begin{proof}
We also prove this by induction on $|V(G)|$.

If $|V(G)| = 1$, then $1 + \cycr{G} = 1 = \lifo{\mvar\ivar}(G)$.

Now suppose $1+\cycr{G'} \geq \lifo{\mvar\ivar}(G')$ for all digraphs $G'$ with $|V(G')| < |V(G)|$.
First we consider the case when $G$ is not strongly connected.  As $|V(H)| < |V(G)|$ for each strongly connected component $H$, by the inductive hypothesis,  there is a monotone $\ivar$-strategy, $\sigma_H$, which captures a fugitive using at most $1+\cycr{H}$ searchers.  From the definition of cycle-rank, for each strongly connected component $H$ of $G$, $\cycr{G} \geq \cycr{H}$, thus $\sigma_H$ uses at most $1+\cycr{G}$ searchers.  We define a monotone $\ivar$-strategy which captures a fugitive on $G$ with at most $1+\cycr{G}$ searchers as follows.  Intuitively, we search the strongly connected components of $G$ in topological order using the monotone strategies $\sigma_H$.  More precisely, let $H_1, H_2, \ldots, H_n$ be an ordering of the strongly connected components of $G$ such that if there is an edge from $H_i$ to $H_j$ then $i < j$.  We define $\sigma$ as follows.  
\begin{itemize}
	\item $\sigma(\epsilon, G) = \sigma_{H_1}(\epsilon, H_1)$,
	\item For $1 \leq i$, if $\letters{X} \subseteq H_i$ and $R = R' \cup \bigcup_{j=i+1}^n H_j$ where $\emptyset \neq R' \subseteq H_i$, $\sigma(X,R) = \sigma_{H_i}(X,R')$,
	\item For $1 \leq i < n$, if $\emptyset \neq \letters{X} \subseteq H_i$ and $R = \bigcup_{j=i+1}^n H_j$ then $\sigma(X,R) = X'$ where $X'$ is the maximal proper prefix of $X$.
\end{itemize}
From the definition of $\ivar$-successors and the ordering of the strongly connected components if $(X_0, R_0), \ldots (X_n, R_n)$ is an $\ivar$-LIFO-search on $G$ where $\letters{X_n} \subseteq H_i$ and $\bigcup_{j>i} H_j \subseteq R_{n-1}   \subseteq \bigcup_{j\geq i} H_j$, then $\bigcup_{j>i} H_j \subseteq R_{n}   \subseteq \bigcup_{j\geq i} H_j$.  
It follows (by induction on the length of a LIFO-search) that every LIFO-search from $(\epsilon, G)$ consistent with $\sigma$ can be divided into a sequence of LIFO-searches $\lambda_1, \lambda_2, \ldots, \lambda_n$, where $\lambda_i$ can be viewed as a LIFO-search consistent with $\sigma_{H_i}$ with $\bigcup_{j > i} H_j$ added to the second component of every position.
Thus if each $\sigma_{H_i}$ is monotone, winning and uses at most $1+\cycr{G}$ searchers, then $\sigma$ is also monotone, winning and uses at most $1+\cycr{G}$ searchers.

Now suppose $G$ is strongly connected.  Let $v_0$ be the vertex which minimizes $f(v) = \cycr{G \setminus \{v\}}$.  Let $G' = G \setminus \{v_0\}$, so $\cycr{G} = 1 + \cycr{G'}$. By the induction hypothesis, there exists a winning monotone $\ivar$-strategy $\sigma'$ which uses at most $1+\cycr{G'}$ searchers to capture a fugitive on $G'$.  We define an $\ivar$-strategy $\sigma$ on $G$ which uses at most $2+\cycr{G'} = 1+\cycr{G}$ searchers as follows.  Initially, place (and keep) a searcher on $v_0$, then play the strategy $\sigma'$ on $G \setminus \{v_0\}$.  More precisely, $\sigma(\epsilon, G) = v_0$ and $\sigma(v_0X,R) = v_0\cdot \sigma'(X,R)$.  Clearly any LIFO-search consistent with $\sigma$ can be viewed as a LIFO-search consistent with $\sigma'$ prepended with the position $(\epsilon, G)$ and where the first component of every position is prepended with $v_0$.  Thus if $\sigma'$ is monotone, then $\sigma$ is monotone, and if $\sigma'$ is winning then $\sigma$ is winning.  Thus $\sigma$ is a monotone winning $\ivar$-strategy which uses at most $1+ \cycr{G}$ searchers.
\end{proof}

\subsection{Relation with other graph parameters}
With a characterization of cycle-rank in terms of several graph searching games we can compare it with other digraph measures defined by similar games.  
In particular, the directed pathwidth of a digraph, $\textrm{dpw}(G)$, which can be characterized by an invisble-fugitive graph searching game~\cite{Bar05}, and the DAG-depth, $\textrm{dd}(G)$ which can be characterized by a visible-fugitive, searcher-stationary searching game~\cite{GHKLOR09}.
Whilst the relationships we present here are known~\cite{Gru08,GHKLOR09}, using the game characterizations we obtain a more simple and more intuitive proof.  
\begin{corollary}
For any digraph $G$, 
\( \textrm{dpw}(G) \leq \cycr{G} \leq \textrm{dd}(G)-1.\)
\end{corollary}

\section{Obstructions for cycle-rank}\label{sec:min-max}

In this section we consider the dual parameter arising from considering the graph searching games from the fugitive's perspective.  We show that it can be characterized by two types of structural features, akin to the havens and brambles used to dually characterize treewidth~\cite{ST93}.  The first of these is the natural generalization of a shelter from~\cite{GT11}, a structural obstruction shown to be dual to tree-depth.

\begin{definition}
A \emph{strong shelter} of a digraph $G$ is a collection $\mathcal{S}$ of non-empty strongly connected sets of vertices such that for any $S \in \mathcal{S}$ 
\[\bigcap \{S' \st S' \in M_{\mathcal{S}}(S)\}  = \emptyset,\]  where $M_{\mathcal{S}}(S)$ is the $\subseteq$-maximal elements of $\{S' \in \mathcal{S}\st S' \subset S\}$.  
The \emph{thickness} of a shelter $\mathcal{S}$ is the minimal length of a maximal $\subseteq$-chain.
\end{definition}

The second structural obstruction we consider is motivated by the definition of a haven in~\cite{JRST01}, a structural feature dual to directed treewidth.

\begin{definition}
A \emph{LIFO-haven of order $k$} is a function $h$ from $V(G)^{<k}$ to induced subgraphs of $G$ such that:
\begin{enumerate}[(H1) ]
\item $h(X)$ is a non-empty strongly connected component of $G \setminus \letters{X}$, and
\item If $X \preceq Y$ and $|Y| < k$ then $h(Y) \subseteq h(X)$.
\end{enumerate}
\end{definition}

Whilst Adler~\cite{Adl05} has shown that the havens of~\cite{JRST01} do not give an exact min-max characterization of directed treewidth and Safari~\cite{Saf05} has shown that directed versions of havens and brambles give rise to distinct parameters, we show that LIFO-havens and strong shelters both give a tight min-max characterization of cycle-rank.


\begin{theorem}[Min-max theorem for cycle-rank]\label{thm:min-max}
Let $G$ be a digraph and $k$ a positive integer.  The following are equivalent:
\begin{enumerate}[(i) ]
\item $G$ has cycle-rank $<k$, 
\item $G$ has no LIFO-haven of order $> k$, and
\item $G$ has no strong shelter of thickness $> k$.
\end{enumerate}
\end{theorem}
\begin{proof}
%
(i) $\Rightarrow$ (ii).  Assume that it is not the case that $G$ has no LIFO-haven of order $>k$, that is, $G$ has a LIFO-haven $h$ of order $k+1$.  We show the fugitive has a winning strategy against $k$ searchers, so by Theorem~\ref{thm:game-char}, $\cycr{G} \geq k$.  Define a $\vscvar$-strategy $\rho$ for the fugitive (against $k$ searchers) by defining $\rho(X,R,X') = h(X')$ for all suitable triples $(X,R,X')$.  From (H1), $(X',\rho(X,R,X'))$ is a valid $\vscvar$-position.  Furthermore, (H2) implies that if $(X,R)$ is a $\vscvar$-position such that $R = h(X)$, then $(X',\rho(X,R,X'))$ is a $\vscvar$-successor of $(X,R)$, so $\rho$ is a $\vscvar$-strategy (defined for all LIFO-searches that use at most $k$ searchers).  Also,  if $(X_0,R_0), (X_1,R_1) \ldots$ is a complete LIFO-search consistent with $\rho$ then $R_i = h(X_i)$ for all $i>0$.  As $h(X) \neq \emptyset$ when $|X|\leq k$, it follows that all consistent complete LIFO-searches that use at most $k$ searchers are winning for the fugitive.  Thus $\rho$ is a winning strategy for the fugitive, so $\lifo{\vscvar}(G) > k$.  By Theorem~\ref{thm:game-char}, $\cycr{G} \geq k$.

(ii) $\Rightarrow$ (iii).  We show that a strong shelter $\mathcal{S}$ of thickness $k$ can be used to define a haven of order $k$.  For each $X \in V(G)^{<k}$ we define $S_X \in \mathcal{S}$ inductively as follows.  For $X = \epsilon$, let $S_\epsilon$ be any $\subseteq$-maximal element of $\mathcal{S}$.  Note that $\{ S \in \mathcal{S}\st S \subset S_{\epsilon}\}$ is a strong shelter of thickness $k-1$.  Now suppose $X = X'v$, $S_{X'}$ is defined, $S_{X'} \cap \letters{X'} = \emptyset$, and $\mathcal{S}_{X'} = \{S \in \mathcal{S}\st S \subset S_{X'}\}$ is a strong shelter of thickness $k-1-|X'|$.  From the definition of a strong shelter, there exists a $\subseteq$-maximal element of $\mathcal{S}_{X'}$ that does not contain $v$, as otherwise $v \in S$ for all $S \in M_{\mathcal{S}}(S_{X'})$.  Let $S_X$ be that element.  As $S_{X'} \cap \letters{X'} = \emptyset$ and $v \notin S_X$, it follows that $S_X \cap \letters{X} = \emptyset$.  Further, $\{S \in \mathcal{S}\st S \subset S_X\}$ is a strong shelter of thickness $(k-1-|X'|)-1 = k-1-|X|$, satisfying the assumptions necessary for the next stage of the induction.  
Now for all $X \in V(G)^{<k}$, $S_X$ is a non-empty strongly connected set such that $S_X \cap \letters{X} = \emptyset$.  Thus there is a unique strongly connected component of $G\setminus \letters{X}$ that contains $S_X$.  Defining $h(X)$ to be that component we see that $h$ satisfies (H1).  For (H2), from the definition of $S_X$, if $X \preceq Y$ and $|Y|<k$, then $S_X \supseteq S_Y$, so $h(X) \supseteq h(Y)$.  Therefore $h$ is a haven of order $k$.

(iii) $\Rightarrow$ (i).  Again, we prove the contrapositive, using a proof similar to~\cite{GT11}.  Suppose $\cycr{G} \geq k$.  Let $G'$ be a strongly connected component of $G$ which has cycle-rank $\geq k$.
We prove by induction on $k$ that $G'$, and hence $G$, has a strong shelter of thickness $k+1$.  Every digraph with $|V(G)|\geq 1$ has a strong shelter of thickness $1$: take $\mathcal{S} = \{\{v\}\}$ for some $v \in V(G)$.  Thus for $k=0$, the result is trivial.  Now suppose for $k'<k$ every digraph of cycle-rank $\geq k'$ contains a strong shelter of thickness $k'+1$.  For $v \in V(G')$, let $G'_v = G' \setminus \{v\}$.  From the definition of cycle-rank, $\cycr{G'_v} \geq k-1$ for all $v \in V(G')$.  Thus, by the induction hypothesis, $G'_v$ contains a strong shelter, $\mathcal{S}_v$, of thickness $(k-1)+1$.  As $v \notin S$ for all $S \in \mathcal{S}_v$, it follows that $\mathcal{S} = \{G'\} \cup \bigcup_{v \in V(G')} \mathcal{S}_v$ is a strong shelter.  As $\mathcal{S}_v$ has thickness $k$ for all $v \in V(G')$, $\mathcal{S}$ has thickness $k+1$.
\end{proof}
\section{Conclusions and further work}\label{sec:conc}
Combining Theorems~\ref{thm:game-char} and~\ref{thm:min-max} gives our main result:
\begin{mainthm}
Let $G$ be a digraph, and $k$ a positive integer.  The following are equivalent:
\begin{enumerate}[(i) ]
\item $G$ has cycle-rank $\leq k-1$,
\item On $G$, $k$ searchers can capture a fugitive with a LIFO-search strategy,
\item On $G$, $k$ searchers can capture a visible fugitive restricted to moving in strongly connected sets with a searcher-stationary search strategy, 
\item $G$ has no LIFO-haven of order $> k$, and
\item $G$ has no strong shelter of thickness $>k$.
\end{enumerate}
\end{mainthm}
This multiple characterization of cycle-rank gives a new perspective on the measure which can be useful for further investigation.  For example, whilst it is known that computing the cycle-rank is NP-complete~\cite{Gru08}, the characterization in terms of a graph searching game with a visible fugitive automatically implies that for any fixed $k$, deciding if a digraph has cycle-rank $k$ is decidable in polynomial time.  From a parameterized complexity perspective, techniques based on separators have shown measures such as directed treewidth are fixed-parameter tractable.  Whether the visible, strongly connected game characterizations of cycle-rank can improve the known complexity from XP to FPT is part of ongoing research.

\bibliographystyle{plain}

\end{document}